\documentclass[reprint,showpacs,amsmath,amssymb,prl,superscriptaddress, floatfix, tightenlines]{revtex4-1}
\usepackage{amsfonts}
\usepackage{amscd}
\usepackage{amsmath,amsthm}    % need for subequations
\usepackage{physics}
\usepackage{enumerate}
\usepackage{epstopdf}
\usepackage{graphicx}
\usepackage{float}
\usepackage{epsfig}
\usepackage{makecell}
\usepackage{thm-restate}
\usepackage{subfigure}
\usepackage{verbatim}
\usepackage[ruled,vlined]{algorithm2e}

\usepackage{algorithmic}
\usepackage{color}

\newtheorem{theorem}{Theorem}

\newtheorem{proposition}{Proposition}
\newtheorem{corollary}{Corollary}

\newcommand{\Q}{\mathsf{Q}}
\newcommand{\A}{\mathsf{A}}
\newcommand{\B}{\mathsf{B}}
\newcommand{\E}{\mathsf{E}}
\renewcommand{\S}{S}
\def\G{G}
\newcommand{\qproof}{\ket{\Q_{\sf proof}}}
\newcommand{\qproofprime}{|\Q_{\sf proof}'\rangle}

\newcommand{\spacespan}{\text{span} }
\usepackage{hyperref}
\usepackage{varioref}

\labelformat{figure}{#1}
\labelformat{algorithm}{(#1)}
\labelformat{theorem}{(#1)}
\labelformat{proposition}{(#1)}
\labelformat{equation}{(#1)}

\hypersetup{colorlinks=true, linkcolor=blue, citecolor=blue, urlcolor=blue  }
\def\test{\mathsf{test}}
\def\prove{\mathsf{prove}}
\def\core{\mathsf{core}}

\newcommand{\sustech}{Department of Physics, Southern University of Science and Technology, Shenzhen 518055, China}

\renewcommand{\eqref}[1]{Eq. (\ref{#1})}

\begin{document}

\title{Verification of Group Non-membership by Shallow Quantum Circuits}

\author{Kai Sun}
\thanks{K.S. and Z.-J.Z. contributed equally to this work.}
\affiliation{CAS Key Laboratory of Quantum Information, University of Science and Technology of China, Hefei 230026, China}
\affiliation{CAS Center For Excellence in Quantum Information and Quantum Physics, University of Science and Technology of China, Hefei 230026, China}

\author{Zi-Jian Zhang}
\thanks{K.S. and Z.-J.Z. contributed equally to this work.}
\affiliation{\sustech}

\author{Fei Meng}
\affiliation{\sustech}
\affiliation{Department of Computer Science, The University of Hong Kong, Pokfulam, Hong Kong SAR, China}

\author{Bin Cheng}
\affiliation{\sustech}
\affiliation{Centre for Quantum Software and Information, Faculty of Engineering and Information Technology, University of Technology Sydney, NSW 2007, Australia}

\author{Zhu Cao}
\affiliation{Key Laboratory of Advanced Control and Optimization for Chemical Processes of Ministry of Education, East China University of Science and Technology, Shanghai 200237, China}

\author{Jin-Shi Xu}
\email{jsxu@ustc.edu.cn}
\affiliation{CAS Key Laboratory of Quantum Information, University of Science and Technology of China, Hefei 230026, China}
\affiliation{CAS Center For Excellence in Quantum Information and Quantum Physics, University of Science and Technology of China, Hefei 230026, China}

\author{Man-Hong Yung}
\email{yung@sustech.edu.cn}
\affiliation{\sustech}
\affiliation{Shenzhen Institute for Quantum Science and Engineering, Southern University of Science and Technology, Shenzhen 518055, China}
\affiliation{Guangdong Provincial Key Laboratory of Quantum Science and Engineering, Southern University of Science and Technology, Shenzhen 518055, China}
\affiliation{Shenzhen Key Laboratory of Quantum Science and Engineering, Southern University of Science and Technology, Shenzhen 518055, China}

\author{Chuan-Feng Li}
\email{cfli@ustc.edu.cn}
\affiliation{CAS Key Laboratory of Quantum Information, University of Science and Technology of China, Hefei 230026, China}
\affiliation{CAS Center For Excellence in Quantum Information and Quantum Physics, University of Science and Technology of China, Hefei 230026, China}

\author{Guang-Can Guo}
\affiliation{CAS Key Laboratory of Quantum Information, University of Science and Technology of China, Hefei 230026, China}
\affiliation{CAS Center For Excellence in Quantum Information and Quantum Physics, University of Science and Technology of China, Hefei 230026, China}

\begin{abstract}
Decision problems are the problems whose answer is either YES or NO. As the quantum analogue of $\mathsf{NP}$ (nondeterministic polynomial time), the class $\mathsf{QMA}$ (quantum Merlin-Arthur) contains the decision problems whose YES instance can be verified efficiently with a quantum computer. The problem of deciding the group non-membership (GNM) of a group element is known to be in $\mathsf{QMA}$. Previous works on the verification of GNM required a quantum circuit with $O(n^5)$ group oracle calls. Here we propose an efficient way to verify GNM problems, reducing the circuit depth to $O(1)$ and the number of qubits by half. We further experimentally demonstrate the scheme, in which two-element subgroups in a four-element group are employed for the verification task. A significant completeness-soundness gap is observed in the experiment. 
\end{abstract}

\maketitle

{\bf Introduction---}
Quantum effect can be used to enhance information processing in many ways. Besides speeding up solving certain problems~\cite{shor1994algorithms,10.1145/780542.780552,nielsen2002quantum}, quantum computers can also be used to construct novel interactive proof systems (IPS)~\cite{goldwasser1989knowledge,watrous2003pspace},
which leads to fruitful studies in blind quantum computing~\cite{broadbent2009universal,fitzsimons2017unconditionally,barz2013experimental}, quantum zero-knowledge proof systems~\cite{broadbent2016zero,grilo2019perfect} and multiprover interactive proof systems ~\cite{ji2020mip,natarajan2019neexp}, etc.
An IPS involves a \emph{verifier} and (potentially multiple) \emph{provers}, where the verifier aims at solving certain problems by exchanging messages with the provers.

IPS can be used to classify decision problems, the problems whose answers can only be YES or NO.
For example, nondeterministic polynomial time ($\mathsf{NP}$), one of the most important complexity classes, can be described by an IPS, with a classical verifier and a single computationally unbounded prover exchanging one round of classical message~\cite{sipser1996introduction,kitaev2002classical}.
Specifically, $\mathsf{NP}$ contains decision problems that, for a YES instance, there exists certain proof message, with which the YES instance can be verified in polynomial time by a classical computer.
$\mathsf{NP}$ can be generalized to the quantum realm naturally and the quantum analogue is called quantum Merlin-Arthur ($\sf QMA$)~\cite{kitaev2002classical, watrous2008quantum}.
In $\mathsf{QMA}$, the proof message is replaced by a quantum state and the verifier can use a quantum computer to process it.

Since a classical verifier can be simulated by a quantum computer and a classical message can be described by a quantum state, every problem belongs to $\mathsf{NP}$ is also in $\mathsf{QMA}$, i.e., $\sf NP \subseteq QMA$.
However, it remains an unsolved problem whether $\sf QMA$ is strictly larger than $\sf NP$ and the group non-membership (GNM) problem is believed to be a possible candidate that falls in $\sf QMA$ but not in $\sf NP$~\cite{group85,watrous2008quantum,babai1992bounded,Watrous2000QuantumProof}. Previous works have shown potential quantum advantage on verifying YES instances of this problem.
It has been proven that the GNM problem is not in $\sf NP^B$~\cite{babai1992bounded} for a certain group oracle $\B$. Also, Watrous proved $\sf GNM(B)\in QMA^B$ for every $\B$ by giving quantum proofs and a verification process which can be efficiently done by a quantum computer~\cite{Watrous2000QuantumProof}.
Furthermore, Watrous conjectured that certain quantum proofs, which is similar to the one constructed for proving $\sf GNM(B)\in QMA^B$, can be used in many other decision problems of finite groups, such as the problems of deciding proper subgroups and simple groups~\cite{Watrous2000QuantumProof}.

Because of the potential applications of quantum IPS and the growing power of near-term quantum devices~\cite{preskill2018quantum,arute2019quantum}, it has become a meaningful question that how to make quantum IPS more friendly for near-term quantum devices.
The verification of the GNM problem is of special importance as it is closely related to the verification of a wide spectrum of group properties and is expected to have quantum advantage.
However, Watrous's process is not favorable for near-term devices as it requires too deep quantum circuits~\cite{Watrous2000QuantumProof,Babai91}.

In this work, we proposed a new verification process which is more friendly to near-term quantum devices based on Watrous's protocol.
The depth of quantum circuit is reduced to $O(1)$ for the groups with at most $2^n$ elements, whereas previous work required $O(n^5)$ oracle calls in one circuit. The number of qubits needed is also half reduced. Our new process makes it easier to use the verification of GNM as a part of near-term quantum applications such as quantum cryptography protocols.
We also demonstrate our new process by an all-optical setup. Various photonic quantum proofs are sent to the optical systems and a significant completeness-soundness gap is presented, showing the validity of our process.

{\bf Group non-membership problem---} First, we formally give the definition of the group non-membership problem here~\cite{Watrous2000QuantumProof}. Let $\G$ be a finite group and ${\S}=\left\langle {{g_1}, \dots ,{g_k}} \right\rangle$ be a subgroup generated by group elements ${{g_1}, \dots ,{g_k}}\in {\G}$.
Given an element $x\in {\G}$, the group non-membership problem is to decide whether $x$ is outside the subgroup $\S$. If $x\notin \S$, $x$ is a YES instance; otherwise, $x$ is a NO instance.

To analyze the problem with minimum assumption on the group, usually the framework of \textit{black-box groups}~\cite{babai1984complexity} is adopted.
In this work, we adopt the same framework as in Watrous's work~\cite{Watrous2000QuantumProof} for the quantum group oracle, in which the quantum group element labels are a set of mutually orthogonal quantum states. We denote the quantum label corresponding to the group element $g$ by $\ket{\psi_g}$ and we denote the space spanned by the quantum labels of elements in $\G$ by $\spacespan\{G\}:=\spacespan\{\ket{\psi_{g_1}}\bra{\psi_{g_2}}:g_1,g_2\in\G\}$. The quantum group oracle is defined to be able to detect whether a state is in $\spacespan\{G\}$ and carry out right multiplication $\mathcal{M}(\cdot)$ as $\mathcal{M}(g_2)\ket{\psi_{g_1}}=\ket{\psi_{g_1 g_2}}$.

In Watrous's process~\cite{Watrous2000QuantumProof}, the quantum proof for the non-membership can be a uniform superposition of the elements in a coset ${\alpha\S}$ of the subgroup $\S$ for any $\alpha \in \G$, where ${\alpha\S}$ is defined as $\alpha \S := \{\alpha s|s\in \S\}$. Explicitly, it can be written as,
\begin{equation}
\label{equ-proofstate}
\qproof = \frac{1}{\sqrt {|\S|}}\sum_{g \in \alpha\S} \ket{\psi_g},
\end{equation}
where $|\S|$ is the element number of the subgroup $\S$. This state is invariant under right multiplications of the elements in $\S$ because they map the elements in $\alpha \S$ bijectively to $\alpha \S$.
On the other hand, if $x\notin\S$, the result state is orthogonal to the original one as $\bra{{\sf Q_{proof}}}\mathcal{M}(x) \left| {{\sf Q_{proof}}} \right\rangle  = 0$ since $(\alpha \S) x$ and $\alpha \S$ do not share common elements.

Next, we introduce the {\it core} quantum circuit that plays a central role in the verification process. The core circuit is similar to the swap test circuit and is depicted in ~\autoref{fig-corecircuit}.
The outcome of the core circuit is defined to be the measurement outcome of the control qubit. We denote by $\core(x,\qproof)=s$ the event of obtaining the measurement outcome $s\in\{0,1\}$ in one run of the core circuit with input state $\qproof$ and group member $x$. The outcome can show the effect of the multiplication by $x$ on the input state.
\begin{figure}[t]
	\centering
	\includegraphics[width=2.9 in]{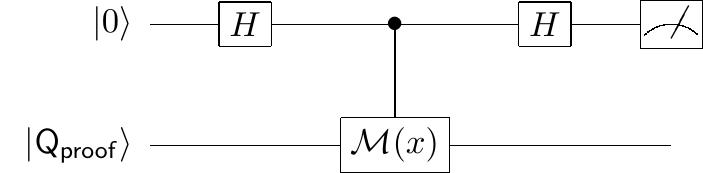}
	\caption{Core circuit. The circuit is similar to the swap test circuit and aims to check whether the input state is invariant under certain group multiplication. With a correct proof state, if $x \in S$, the measurement outcome is always 0; if $x \notin S$, the measurement outcome is 1 with probability 0.5.}
	\label{fig-corecircuit}
\end{figure}
For $\core(x,\qproof)$, if $x \in {\S}$, the outcome can only be $0$ as $\qproof$ is invariant under the multiplication. If $x \notin {\S}$, the probability of obtaining $1$ is $0.5$ as the state after multiplication is orthogonal to $\qproof$. Therefore, with the proof state, the non-membership of an element can be verified when the outcome $1$ is obtained.

However, a malicious prover may send bogus proof states that deviate from \autoref{equ-proofstate} and give incorrect outcomes. Therefore, to ensure the soundness of the verification, the verifier has to do a property checking on the received proof state, i.e., check the state is invariant under the group multiplication $\mathcal{M}(s)$ for any $s\in \S$, so that the elements in $\S$ cannot be proven to be outside $\S$. In the original process~\cite{Watrous2000QuantumProof}, to verify the proof received is valid, the verifier needs to uniformly sample the subgroup elements in a reversible way and produce a quantum superposition of the quantum labels
\begin{equation}
\label{equ-quantumsuperpostion}
\sum_{g\in S} a_i \ket{g}\ket{\text{garbage}(g)},
\end{equation}
where $\{a_i\}$ should be nearly uniform.
However, the reversible sampling requires $O(n^5)$ calls of the group oracle in the quantum circuit according to Remark 8.3 in \cite{Babai91} and requires the verifier to keep at least one more quantum group element label. 
%Therefore, the original process is difficult to be realized in near-term devices.

\textbf{Simplified verification---}
In this work, we reduce both the circuit depth and qubit number needed for the verification of GNM.
The simplification relies on the technique which we call \textit{Random State Inspection} (RSI). RSI can be used in verification processes that includes a property checking phase of the proof state. Usually, in these processes the property checking and the after verification (verification after property ensured) are done in one quantum circuit. Watrous's process is such a process. In Watrous's process, the verifier first ensures that the received proof state is invariant under multiplication with subgroup members by multiplying it with the state in \autoref{equ-quantumsuperpostion} and then carries out the verification with the core circuit. RSI provides a way to reduce the circuit depth by separating the property checking and the verification after property checking.

In RSI, the prover is required to send $m$ registers that carry copies of a state to the verifier. The verifier randomly selects one register to reserve and apply independent test channels to the other $m-1$ registers to check the property of the states that they carry. If all the $m-1$ registers pass the property checking, the verifier accepts the reserved register for the later verification process. Otherwise, the verifier rejects. 
We prove that, if all the other registers have passed the test channel, the probability for the reserved register to fail passing the test channel if tested can be bounded to $0$ at speed $O(1/m)$ even when the $m$ registers are entangled. The verifier can then directly apply the after verification on the reserved register as its property is ensured. By RSI, the circuit depth needed in the verification is reduced to that of the property checking or the after verification.

More important, we simplify the property checking process for the proof state in GNM. Rather than using the state in \autoref{equ-quantumsuperpostion} which needs $O(n^5)$ quantum group operations to produce, we propose a test channel in which a subgroup element $s$ is first classically sampled from a nearly uniform distribution by Babai's algorithm~\cite{Babai91}, followed by checking whether $\core(s,\rho)=0$. Here, `nearly uniform' means the probability for $s$ to be any subgroup elements is in $(1/|S| - 1/2^{2n}, 1/|S| + 1/2^{2n})$.
We denote the probability for a state $\rho$ to pass the test channel by $\Pr(\core(s,\rho)=0)$. We prove that for any element $g\in \S$ and any quantum state $\rho\in\spacespan\{G\}$, the probability of incorrectly proving the non-membership of $g \in S$, i.e. having $\core(g,\rho)=1$, can be bounded as,
\begin{equation}
    \Pr(\core(g,\rho)=1)\leq 4 \left( 1-\Pr(\core(s,\rho)=0) \right) \ .
\end{equation}
By RSI, we can ensure $\Pr(\core(s,\rho)=0)$ is high enough and therefore bound the error probability $\Pr(\core(g,\rho)=1)$.

To summarize, in our new process, we split the property checking of the proof state and the verification after property checking into different circuits by RSI. We also use a new property checking process which requires much less quantum resources.
As a result, the verifier only needs to run the core circuits, which is shallow, for many times, rather than run a deep circuit with $O(n^5)$ group operations.
Also, the number of qubits that the verifier needs to keep is halved because the verifier no longer needs the keep the state in \autoref{equ-quantumsuperpostion}. 
Detailed and rigorous analysis can be found in the supplementary material.

\begin{figure*}[t]
\centering
\includegraphics[clip=true, width=\linewidth]{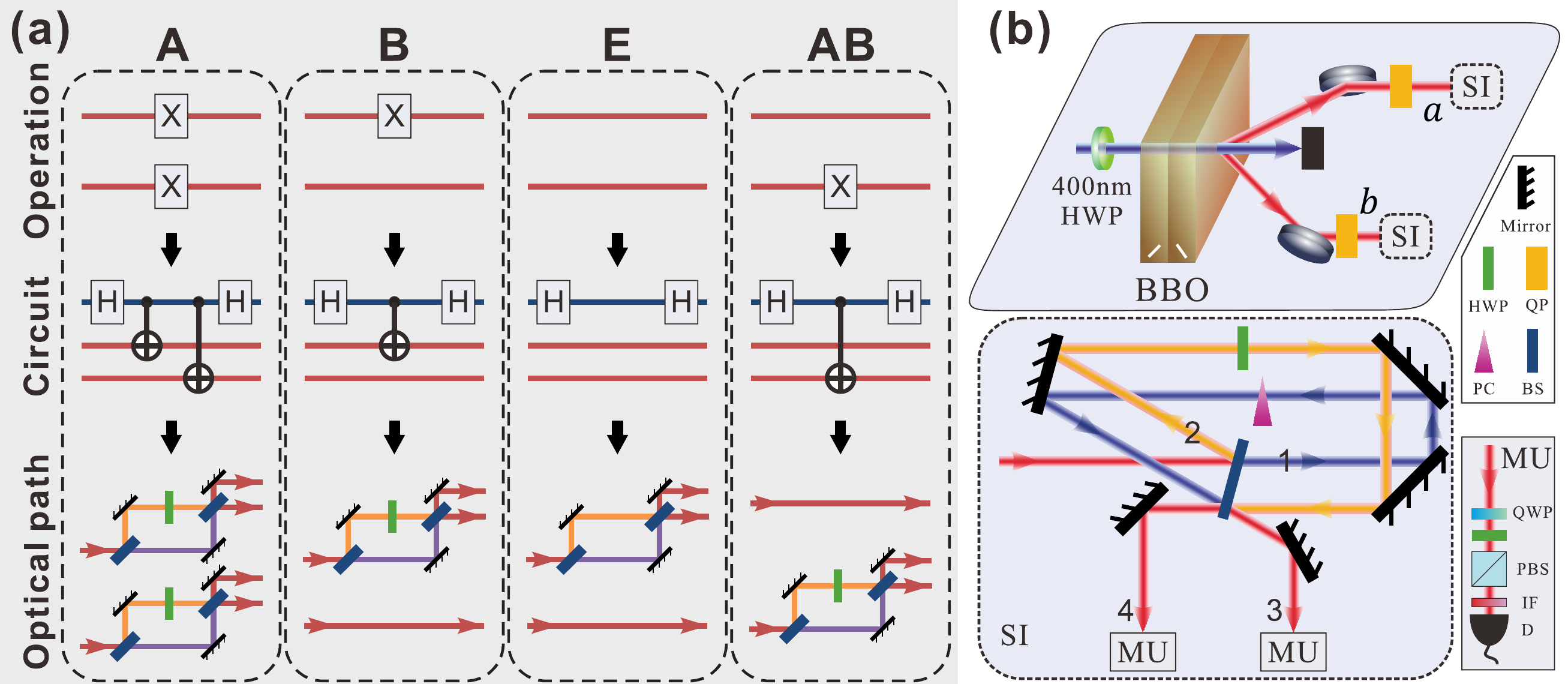} %
\caption{(Color online). Circuit mapping and experimental setup. {\bf (a)} The circuits for group multiplications in the first line are deduced from the quantum labels for the elements and can be easily proven. They are modified to its controlled version and used to construct the core circuits for the verification process in the second line. Optical paths are presented in the third line. Here, two beam splitters (BSs) building a Mach-Zehnder interferometer (MZI) are used to play the role as two Hadamard operations on the control qubit which is realized with the path information. One path is regarded as $\ket{0}$ and the other one is $\ket{1}$. A half wave plate (HWP) is placed in $\ket{1}$ path to act the CNOT gate on the polarization qubit with the optical axis at $45^\circ$. {\bf (b)} Experimental setups. Entangled photon pairs are produced by pumping BBO and using quartz plates (QP) on the above panel. Two photons are sent to the sides $a$ and $b$ respectively. On each side, a Sagnac interferometer (SI) shown on the bottom panel in detail is constructed to realize the MZI. In a SI, a HWP is placed in $\ket{1}$ path (shown in orange beam and marked as 2) and a phase compensation (PC) crystal is located in $\ket{0}$ (shown in blue beam and marked as 1). Measurement unit (MU) consisting of a quarter wave plate (QWP), a HWP, a polarization beam splitter (PBS) and a single photon detector equipped with an interferometer filter (IF) is placed on each output port (marked as 3 and 4) of SI. Note, in this figure, unitary of multiplying by $\A$ is realized. With removing the SI, we can implement different quantum circuits.			}\label{fig:setup}
\end{figure*}

{\bf Experimental setup---}In this work, an experimental demonstration of the our new verification process is carried out. We consider an abelian group $\sf G= \{\langle A,B\rangle|AB=BA, A^2=B^2=E\}.$
The four elements in $\sf G$ are encoded in the polarization degree of freedom of photons as $ \ket{\psi_\mathsf{E}} = \ket{VH},  \ket{\psi_\mathsf{A}} = \ket{HV},  \ket{\psi_\mathsf{B}} = \ket{HH}$ and $\ket{\psi_\mathsf{AB}})= \ket{VV}$, in which $|H\rangle$ and $|V \rangle$ denote the horizontal and vertical polarization, respectively. They can together span the whole two-qubit Hilbert space.
The optical realization of the controlled right multiplication of $\{\mathsf{A},\ \mathsf{B},\ \mathsf{AB},\ \mathsf{E}\}$ are illustrated in \autoref{fig:setup}\textbf{(a)}. The subgroups we choose are $\sf S = \left\{ \sf E, A\right\}$ and $\mathsf{S'} = \left\{ \sf E, AB\right\}$.
The quantum proof states for $\mathsf{S}$ and $\mathsf{S'}$ used in the experiment are
$
\label{Q_proof_exp}
\qproof = \frac{1}{\sqrt{2}}(\ket{\psi_\mathsf{B}}+\ket{\psi_\mathsf{AB}}) =\frac{1}{\sqrt{2}}(\ket{HH}+\ket{VV}),
$
and
$
\label{Q_proof_exp}
\qproofprime = \frac{1}{\sqrt{2}}(\ket{\psi_\mathsf{B}}+\ket{\psi_\mathsf{A}}) =\frac{1}{\sqrt{2}}(\ket{HH}+\ket{HV}),
$
respectively.

In the experimental, we put $\qproof, \qproofprime, \ket{\psi_\mathsf{A}}$ and $\ket{\psi_\mathsf{B}}$ in the core circuit with right multiplication by $\sf E,A,B$ and $\sf AB$. The full experimental setup is shown in \autoref{fig:setup}\textbf{(b)}. The input states are generated by pumping two identically cut type-I beta-barium-borate (BBO) crystals whose optic axes are aligned in mutually perpendicular planes~\cite{Kwiat99} with an ultraviolet (UV) source. The UV pulses is frequency doubled from a mode-locked Ti:sapphire laser centered at 800 nm with 130 fs pulse width and 76 MHz repetition rate. After compensating the birefringence effect between $H$ and $V$ in BBO crystals with quartz plates (QP), maximally entangled photon pairs of the forms $\qproof=(| HH \rangle + | VV \rangle)/\sqrt{2}$ are produced~\cite{PhysRevA.60.R773}.
Furthermore, by adjusting the polarization of pump pulses and down-conversion photons, the other states of $\ket{\psi_\mathsf{B}}=|HH\rangle$, $\ket{\psi_\mathsf{A}}=|HV\rangle$ and $\qproofprime=(| HH \rangle + | HV \rangle)/\sqrt{2}$ are produced.
The input photons are then sent to one of the quantum circuits in \autoref{fig:setup} to perform the core circuit with different group multiplications.
In our setup, the Mach-Zehnder interferometer is realized by Sagnac interferometer in which the path information of photons is regarded as the control qubit~\cite{sun2018}. In a Sagnac interferometer, an optical non-polarization beam splitter (BS), worked as the Hadamard gates on control qubit, is used to separate the beam into two paths $1$ and $2$ which are treated as the control qubit $|0\rangle$ and $|1\rangle$ respectively. Here, the BS is chosen to split $50:50$ for $0^\circ$ angle of incidence which could be decrease the difference of split ratio of different polarizations. In the path $|1\rangle$, a half-wave plate (HWP) is used to implement CNOT gates with set at $45^\circ$ to reverse the photon polarization. The visibilities of two Sagnac interferometers are $96.7\pm0.4\%$ and $95.9\pm0.4\%$, respectively. Note that, for the circuit $\sf E$ with $\sf E$ multiplication, there is no CNOT gate and the HWP is set at $0^\circ$. Beams $1$ and $2$ combine in the BS and then are separated as beams $3$ and $4$. The polarization of photons are analyzed on the outputs of beams $3$ and $4$ by polarization beam splitters (PBS), HWP and quarter-wave plates (QWP). The photons are detected by single photon detectors ($D$) equipped with 3 nm interference filters (IF).

For the circuit $\mathsf{A}$, which implements the multiplication by $\mathsf{A}$, the probability $P_0$ of detecting $\ket{0}$ equals the sum of coincidence count (CC) of detectors located at $a3$ and $b3$ and CC of detectors located at $a4$ and $b4$, where $a3$ is the output port 3 of the SI on the side of $a$, and similarly hereinafter. 
The probability $P_1$ of detecting $\ket{1}$ equals the sum of CC of $a3$ and $b4$ and CC of $a4$ and $b3$.
On the other hand, for the case of circuit $\mathsf{B}$, according to the corresponding mapping relation where the Sagnac interferometer is only placed in the $a$ side, the probability of detecting $\ket{0}$ equals the CC of $a3$ and $b$, and the probability of $\ket{1}$ equals the CC of $a4$ and $b$. The similar methods suit the other circuits $\mathsf{AB}$ and $\mathsf{E}$.

Besides the interference visibility introduced above, two Sagnac interferometers are further verified with the input state of $(\ket{HH}+\ket{VV})/\sqrt{2}$ which is prepared with a fidelity of $95.9\pm1.0\%$. For the Sagnac interferometer appearing in the $\E$ circuit, the output state generated from the CC of $a3$ and $b$ remains the maximally entangled state and is achieved experimentally with a fidelity of $95.3\pm1.0\%$. For the other interferometer which is used in circuit $\mathsf{AB}$, without inserting the CNOT gate, the output state generated from the CC of $a$ and $b3$ is also the same with the input state and achieves a fidelity of $94.2\pm1.4\%$. We further verify other output cases of the interferometers and achieve high fidelities for them. The real and imaginary parts of all corresponding density matrices are presented in the supplementary material.

{\bf Experimental results---}
\begin{figure}[t]
	\centering
	\includegraphics[width= 3.4in]{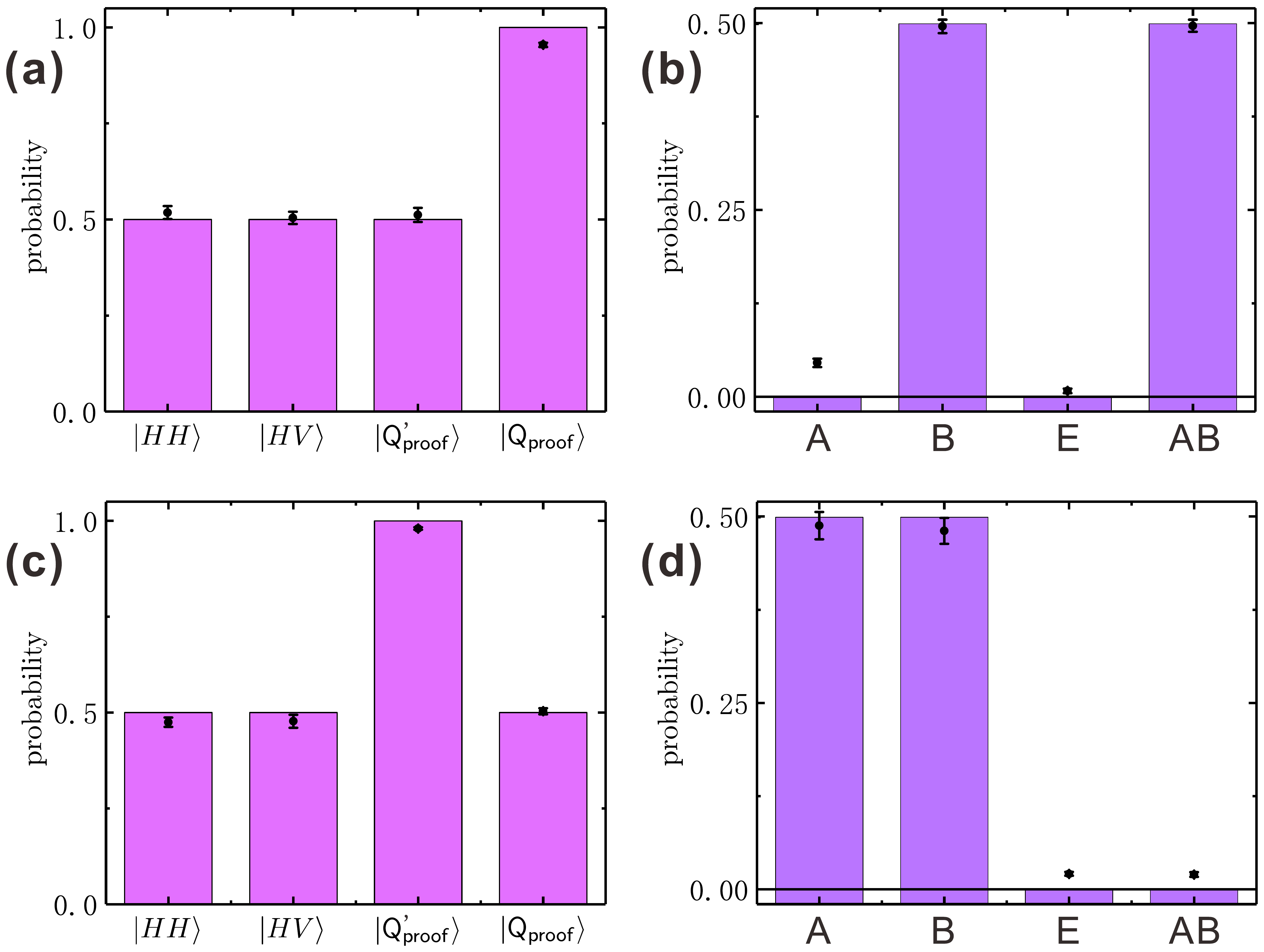}
	\caption{(Color online). 
	Experimental results. {\bf (a)}, {\bf (c)} The detecting probabilities of $\ket{0}$ with different input proof states for the circuit $\sf A$ and $\sf AB$, respectively. {\bf (b)}, {\bf (d)} The probability for the proof states $\qproof=\ket{HH}+\ket{VV})/\sqrt{2}$ and $\qproofprime=(\ket{HH}+\ket{HV})/\sqrt{2}$ to prove group non-membership for every group elements (detecting the control qubit in $\ket{1}$). The histograms and black points are theoretical and experimental results, respectively. All error bars are estimated to be standard deviation from the statistical variation of the photon counts assumed to follow a Poisson distribution}
		\label{fig:ExpRes}
\end{figure}
Equipped with the experiment setup, we first carried out our new process on the group $\sf S$.
To demonstrate the test channel in our verification process, the correct proof $\qproof$ as well as the bogus proofs $\qproofprime, \ket{\psi_\mathsf{A}}$ and $\ket{\psi_\mathsf{B}}$ are produced and sent to the core circuit with multiplication by $\A$. A state passes the test if the control qubit is detected to be in $\ket{0}$.
The results are shown in \autoref{fig:ExpRes}{\bf(a)}. We find that the probabilities for a bogus proof to pass the state test do not exceed $0.518\pm0.017$ and have a significant gap towards the probability $0.955\pm0.006$ for a correct proof state $\qproof$ to pass.
Then we show how the non-membership of an element $g$ can be verified with the correct proof state $\qproof$. The group non-membership of $g$ is verified when $\ket{1}$ is detected in the core circuit with multiplication by $g$. 
The experimental result is shown in \autoref{fig:ExpRes}{\bf(b)}.
We find the probabilities for $\qproof$ to be accepted is higher than $0.496\pm0.009$ for $\sf B,AB\notin \mathsf {S'}$ and lower than $0.045\pm0.006$ for  $\sf E,A\in \mathsf {S'}$.
The above analysis implies that if the prover sends $m$ registers and the verifier chooses $m-1$ registers to test, the probability $p_c$ for a group of correct proof state to be accepted is greater than $0.496 (0.949)^{m-1}$. In contrast, the probability for $m$ bogus state to pass the tests is only $(0.518)^{m-1}$ .
For a general bogus proof, our theory shows that the probability $p_s$ for it to be accepted is bounded by $\frac{16}{7(m-1)}$. Therefore, the gap $p_c-p_s$ is maximized when $m=14$ and the maximal value is $0.075$.
	
For the other subgroup $\mathsf{S'}$, the result is similar. The bogus proofs become $\ket{\psi_\mathsf{A}}$, $\ket{\psi_\mathsf{B}}$, $\qproof$ and the correct proof becomes $\qproofprime$. The probabilities for the bogus proofs to pass the test channel do not exceed  $0.503\pm0.008$ and for the correct proof $\qproofprime$, the corresponding probability is $0.980\pm0.003$ as shown in \autoref{fig:ExpRes}{\bf(c)}. We find that the probability for $\qproofprime$ to be accepted when used for verifying GNM is higher than $0.481\pm0.017$ for $\sf B,A\notin \mathsf {S'}$ and is lower than $0.020\pm0.003$ for $\sf E,AB\in \mathsf {S'}$ as shown in \autoref{fig:ExpRes}{\bf(d)}. In this case, $p_s$ is still bounded by $\frac{16}{7(m-1)}$ and $p_c=0.481(0.980)^{m-1}$. The gap $p_c-p_s$ is maximized when $m=19$ and the maximal value is $0.207$. These completeness-soundness gaps indicates the success of our experiment.

{\bf Conclusion---} In this work, a new quantum verification process for the GNM problem is proposed, in which the required quantum resources are greatly reduced. We experimentally demonstrate the verification scheme in an all-optical setup. Our novel verification process can be used to construct more quantum protocols for near-term quantum devices. Furthermore, as it is very likely that similar verification process of GNM can be used in other problems of finite groups, it will be interesting if this validity was formally proven and experimentally demonstrated.

\subsection*{Acknowledgments}
This work was supported by the National Key Research and Development Program of China (Grants NO. 2016YFA0302700 and 2017YFA0304100), National Natural Science Foundation of China (Grant NO. 11821404, 11774335, 61725504, 61805227, 61975195, U19A2075), Anhui Initiative in Quantum Information Technologies (Grant NO.\ AHY060300 and AHY020100), Key Research Program of Frontier Science, CAS (Grant NO.\ QYZDYSSW-SLH003), Science Foundation of the CAS (NO. ZDRW-XH-2019-1), the Fundamental Research Funds for the Central Universities (Grant NO. WK2030380017, WK2030380015 and WK2470000026). 
Natural Science Foundation of Guangdong Province (Grant NO.2017B030308003), the Key R\&D Program of Guangdong province (Grant NO. 2018B030326001), the Science, Technology and Innovation Commission of Shenzhen Municipality (Grant NO. JCYJ20170412152620376 and NO. JCYJ20170817105046702 and NO. KYTDPT20181011104202253), National Natural Science Foundation of China (Grant NO. 11875160 and NO. U1801661), the Economy, Trade and Information Commission of Shenzhen Municipality (Grant NO.201901161512), Guangdong Provincial Key Laboratory (Grant NO. 2019B121203002).

Z.-J.Z. proposed and proved the theory; K.S. conducted the experiments.

	% choose a style
	%\bibliographystyle{ieeetr}
	\bibliographystyle{apsrev4-1}
	\bibliography{main}
	%\nocite{*}
	\clearpage
    \onecolumngrid
	\begin{appendix}
		\section*{Supplementary Material}
		For conciseness, in the supplementary material we use $\ket{g}$ rather than $\ket{\psi_g}$ to denote the quantum label of group element $g$.
		First we formally give the definition of Random State Inspection (RSI) and the test channel used in our new process here.
\begin{algorithm}\caption{Random State Inspection\label{alg-general-test}}
\begin{enumerate}
	\item Prover sends $m$ registers to the verifier. The registers should carry the same quantum state.
	\item Verifier uniformly randomly reserves one register and applies a test channel $T$ to the other $m-1$ registers. The test channel $T$ should map each register to a one-qubit state, with $\ket{0}$ indicating \textit{pass} and $\ket{1}$ indicating \textit{fail}.
	\item Verifier measures the outputs of the test channels in the computational basis. The reserved register is accepted if and only if all outcomes are 0.
\end{enumerate}
\end{algorithm}

\begin{algorithm}\caption{Proof State Test Channel $T$}
\begin{enumerate}
	\label{alg-proof-test}
	\item Verifier uses the group oracle to check whether the state of the register is in $\spacespan\{G\}$. If it is the case, output $\ket{1}$.
	\item Verifier randomly samples a subgroup element by Babai's algorithm~\citep{Babai91} by a classical computer, with the probability of sampling each element in $(\frac{1}{|S|}-\frac{1}{2^{2n}},\frac{1}{|S|}+\frac{1}{2^{2n}})$. Denote the element sampled by $s$.
	\item Verifier applies the core circuit to the register with element $s$ by the group oracle.
	\item Verifier gives the control qubit in the core circuit as the output.
\end{enumerate}
\end{algorithm}
		
With \autoref{alg-general-test} and \autoref{alg-proof-test}, we summarize our new verification process here.
		\begin{algorithm}\caption{Verification of Group Non-membership}
			\begin{enumerate}
				\label{alg-qma_veri}
				\item Prover sends $m$ registers with state in \autoref{equ-proofstate} to the verifier, trying to prove that the element $g$ in the group $G$  is not in the subgroup $S$ of $G$. There are at most $2^n$ elements in $G$.
				\item (The $\test$ subroutine) Verifier runs \autoref{alg-general-test} with test channel $T$ defined in \autoref{alg-proof-test} and obtains the reserved register.
				\item (The $\prove$ subroutine) Verifier applies the core circuit with $g$ to the reserved register after checking the state of the reserved register is in the space $\spacespan\{G\}=\spacespan\{\ket{g_1}\bra{g_2}:g_1,g_2\in\G\}$ (the space spanned by the valid labels) by the group oracle. The prover passes this subroutine if the outcome of the core circuit is $1$ and the state of the register is found in the valid space $\spacespan\{G\}$.
				\item The prover passes the verification process if he passes both the $\test$ and the $\prove$ subroutine.
			\end{enumerate}
		\end{algorithm}
		
		Note that RSI can achieve its goal even when the registers are mixed and entangled. Therefore, a quantum proof $\rho_s$ can be seen as carried by the registers. To prove the validity of \autoref{alg-qma_veri}, we need to prove its completeness and soundness. We first prove the completeness of it.
		\begin{theorem}[Completeness]
			\label{thm-completeness}
			In \autoref{alg-qma_veri}, if the prover and the verifier are honest, for any group element $g$ not in the subgroup, the probability for the prover to prove the non-membership of $g$ is $\frac{1}{2}$.
		\end{theorem}
		\begin{proof}
			In the $\test$ subroutine, the honest prover will always pass the test because for a honest proof state $\qproof$ $\Pr(\core(s,\qproof))=1$. Then, in the $\prove$ subroutine, the copy used in this phase has a probability of $\frac{1}{2}$ to be accepted. This can be proved by observing how the proof state is transformed by the core circuit. The transformation is
			\begin{equation}
			\begin{split}
		&\frac{1}{\sqrt {|\S|}}\sum_{s \in \alpha\S} \ket{s} \rightarrow \frac{1}{2\sqrt {|\S|}} \left(\ket{0} \sum_{s \in \alpha\S} (\ket{s}+\ket{s g}) +
			\ket{1} \sum_{s \in \alpha\S} (\ket{s}-\ket{s g})\right),
			\end{split}
			\end{equation}
			where the control qubit is put in the left. The probability to obtain $\ket{1}$ after measuring the control qubit is
			\begin{equation}
			    p=\left\Vert \frac{1}{2\sqrt {|\S|}}\sum_{s \in \alpha\S} (\ket{s}-\ket{s g})\right\Vert^2.
			\end{equation}
			Because when $g\notin S$, we have $sg\notin \alpha S$ and because for any two $s_1,s_2\in \alpha S, s_1\neq s_2$, we have $s_1g\neq s_2g$, the norm can be easily calculated and we can obatin $p=\frac{1}{2}$. Therefore, the probability to prove the non-membership of an element $g\notin S$ is $\frac{1}{2}$ by \autoref{alg-qma_veri}. 
			
			All the verification processes can be done in polynomial time with a quantum computer. According to \autoref{thm-soundness}, only polynomial number of registers are required. The process of every register is only the core circuit with a random group element. The sampling of random group element can also be efficiently done according to \autoref{thm-babai-elementsampler}.
		\end{proof}
		
		Below, we denote $\rho_r$ as the density matrix of the reserved register. Also, we denote $(\test(\rho_s) = 1)$ as the event that the prover passes this subroutine and $(\test(\rho_s) = 0)$ otherwise. In addition, we denote $(T(\rho_r) = 1)$ as $\rho_r$ passes the test channel (not the $\test$ subroutine). Similarly, we denote $(\prove(g,\rho_s) = 1)$ as the event that the proof state passes the verification process, and $(\prove(g,\rho_s) = 0)$ otherwise.

		\begin{theorem}[Soundness]
			\label{thm-soundness}
			For any $g\in S$, the probability to incorrectly verify its non-membership using \autoref{alg-qma_veri} vanishes as,
			\begin{equation}
			\Pr(\prove(g,\rho_s)= 1) \leq \frac{8}{m} \ ,
			\end{equation}
			for any proof state $\rho_s$.
		\end{theorem}
		\begin{proof}
		    In this proof we assume that the quantum proof sent by the prover is always in the valid space $\spacespan\{G\}=\spacespan\{\ket{g_1}\bra{g_2}:g_1,g_2\in\G\}$ because in \autoref{alg-qma_veri} the verifier always check whether the states is in $\spacespan\{G\}$ before doing any operation on them. After the checking, the quantum proof is either projected to $\spacespan\{G\}$ or rejected. Therefore, the prover making the quantum proof not in $\spacespan\{G\}$ will only decrease the probability for it to be accepted.
		
			To prove the non-membership of $g$, the registers should first pass the $\test$ subroutine. By \autoref{thm-veripassandtestpass}, the probability $\Pr(\test(\rho_r)=1)$ for the reserved copy $\rho_r$ to pass the test is bounded as
			\begin{equation}
			\Pr(T(\rho_r)=1|\test(\rho_s)=1)\geq 1-(\frac{1}{\Pr(\test(\rho_s) = 1)}-1)\frac{1}{m-1} \ .
			\end{equation}
			By \autoref{thm-pass-soundness}, we know that the the high probability to pass the test means low error probability to prove the non-membership. For the reserved copy, provided that the $\test$ subroutine has been passed, the probability of incorrectly proving GNM by it can be given by,
			\begin{equation}
			\Pr(\core(g,\rho_r) = 1 | \test(\rho_s) = 1) \leq \frac{1-\Pr(T(\rho_r)=1|\test(\rho_s)=1)}{K(1-\frac{|S|}{2^{2n}})} \ ,
			\end{equation}
			where
			\begin{equation}
			\frac{1}{K}=1-\cos(\left\lceil \frac{|g|}{2} \right\rceil \frac{2}{|g|}\pi)\leq 2 \ .
			\end{equation}
			Combining the equations, we can obtain the probability of incorrectly accepting the proof state
			\begin{equation}
			\label{equ-soundness-last}
			\begin{split}
			&\Pr(\test(\rho_s) = 1, \core (g,\rho_r) = 1 )\\
			=& \Pr(\test(\rho_s) = 1) \cdot \Pr(\core (g,\rho_r) = 1 | \test(\rho_s) = 1) \\
			\leq&  \frac{\Pr(\test(\rho_s) = 1)}{K(1-\frac{|S|}{2^{2n}})} (1-\Pr(T(\rho_r)=1 | \test(\rho_s) = 1)) \\
			\leq&  \frac{\Pr(\test(\rho_s) = 1)}{K(1-\frac{|S|}{2^{2n}})} (\frac{1}{\Pr(\test(\rho_s) = 1)}-1)\frac{1}{m-1} \\
			\leq & \frac{1}{K}(1+\frac{1}{2^{n+1}-1})\frac{1}{m-1} \\
			\leq & \frac{8}{m} \ .
			\end{split}
			\end{equation}
			
		\end{proof}
		
		\begin{corollary}
		For the group $\sf G= \{\langle \A,\B\rangle|\A\B=\B\A, \A^2=\B^2=\E\}$ used in the experimental demonstration in the main text, the probability to incorrectly verify the non-membership of a subgroup element $g$ except $\E$ using \autoref{alg-qma_veri} vanishes as,
			\begin{equation}
			\Pr(\prove(g,\rho_s)= 1) \leq (\frac{16}{7})\frac{1}{m-1} \ ,
			\end{equation}
			for any proof state $\rho_s$.
        \end{corollary}
		
		\begin{proof}
		The rank of the group elements in $\sf G= \{\langle \A,\B\rangle|\A\B=\B\A, \A^2=\B^2=\E\}$ are all $2$ except $\E$. Therefore $1/K=2$. Notice that there are $4$ elements in $\sf G$ and $n=2$, refer to \autoref{equ-soundness-last}, we have
		\begin{equation}
		    \Pr(\test(\rho_s) = 1, \core (g,\rho_r) = 1 )\leq \frac{1}{K}(1+\frac{1}{2^{n+1}-1})\frac{1}{m-1} = 2(1+\frac{1}{2^{2+1}-1})\frac{1}{m-1} = (\frac{16}{7})\frac{1}{m-1}.
		\end{equation}
		\end{proof}
		
		\begin{proposition}
			\label{thm-veripassandtestpass}
			In \autoref{alg-general-test}, the probability $\Pr(T(\rho_r)=1|\test(\rho_s)=1)$ for the reserved register $\rho_r$ to pass the test in the condition that all other registers have passed the tests can be bounded as
			\begin{equation}
			\Pr(T(\rho_r)=1|\test(\rho_s)=1)\geq 1-(\frac{1}{\Pr(\test(\rho_s)=1)}-1)\frac{1}{m-1},
			\end{equation}
			where $\Pr(\test(\rho_s)=1)$ is the overall probability for the registers to pass \autoref{alg-general-test}.
		\end{proposition}
		\begin{proof} After applying $T$ to every register in $\rho_s$, we can obtain an $m$-qubit state as
			\begin{equation}
			T^{\otimes M} (\rho_s)=\sum_{s_1,s_2\in\{0,1\}^m} a_{s_1,s_2} \ket{s_1}\bra{s_2} \ .
			\end{equation}
			Then we try to bound the probability for passing RSI. For the sake of conciseness, we denote $a_{s,s}$ by $a_s$. Also, we denote the $a_s$ in which $s$ is $0$ in every indices except in index $i$ as $a_{i=1}$. We have
			\begin{equation}
			\begin{split}
	\Pr(\test(\rho_s)=1)&=\frac{1}{m}\sum_{i=1}^m (a_{0^m}+a_{i=1})=a_{0^m}+\frac{1}{m}\sum_{i=1}^m a_{i=1}\\
			&\leq  a_{0^m}+\frac{1}{m}(1-a_{0^m}),
			\end{split}
			\end{equation}
			and
			\begin{equation}
			a_{0^m}\geq \frac{m\Pr(\test(\rho_s)=1)-1}{m-1}.
			\end{equation}
			Decompose the density matrix $\rho_r$ of the reserved register after being applied with channel $T$ as
			\begin{equation}
			\rho_r=\sum_{s_1,s_2\in\{0,1\}^2} r_{s_1,s_2} \ket{s_1}\bra{s_2}.
			\end{equation}
			According to the definition of the conditional probability $\Pr(A|B) = \Pr(A\cap B)/\Pr(B)$, we have
			
			\begin{equation}
			\begin{split}
			&\Pr(T(\rho_r)=1|\test(\rho_s)=1)=\frac{a_{0^m}}{\Pr(\test(\rho_s)=1)}\\
			\geq & 1-(\frac{1}{\Pr(\test(\rho_s)=1)}-1)\frac{1}{m-1}.
			\end{split}
			\end{equation}
			
		\end{proof}

        \begin{corollary}
        After \autoref{alg-general-test}, if the reserved register is tested, the probability $\Pr(T(\rho_r)=0, \test(\rho_s)=1)$ for the reserved register to fail passing the test channel together with that the RSI is not rejected is less than $O(1/m)$.
        \end{corollary}
		
		\begin{proof}
		By \autoref{thm-veripassandtestpass} we know that
		\begin{equation}
			\Pr(T(\rho_r)=1|\test(\rho_s)=1)\geq 1-(\frac{1}{\Pr(\test(\rho_s)=1)}-1)\frac{1}{m-1}.
			\end{equation}
			Easily we can know that
			\begin{equation}
			\Pr(T(\rho_r)=1|\test(\rho_s)=1)\Pr(\test(\rho_s)=1)\geq \Pr(\test(\rho_s)=1)-(1-\Pr(\test(\rho_s)=1))\frac{1}{m-1},
			\end{equation}
			\begin{equation}
			\Pr(T(\rho_r)=1,\test(\rho_s)=1)\geq \Pr(\test(\rho_s)=1)-(1-\Pr(\test(\rho_s)=1))\frac{1}{m-1},
			\end{equation}
			\begin{equation}
			\Pr(T(\rho_r)=1,\test(\rho_s)=1)\geq \Pr(\test(\rho_s)=1)-O(1/m),
			\end{equation}
			\begin{equation}
			\Pr(T(\rho_r)=1,\test(\rho_s)=1)\geq\Pr(T(\rho_r)=1,\test(\rho_s)=1)+\Pr(T(\rho_r)=0,\test(\rho_s)=1)-O(1/m),
			\end{equation}
			\begin{equation}
			\Pr(T(\rho_r)=0,\test(\rho_s)=1)\leq O(1/m).
			\end{equation}
		\end{proof}

		\begin{restatable}{proposition}{thm-pass-soundness}
			\label{thm-pass-soundness}
			For any state $\rho_r\in\spacespan\{\ket{g_1}\bra{g_2}:g_1,g_2\in\G\}$, the probability $\Pr(T(\rho_r)=1)$ for it to pass the test in \autoref{alg-proof-test} and the probability $\Pr(\core(g,\rho_r)=1)$ for a group element $g\in S$ to be proven not in $S$ by $\rho_r$, has the relation that
			\begin{equation}
			\label{equ-prove-pass}
			\Pr(\core(g,\rho_r)=1)\leq \frac{1-\Pr(T(\rho_r)=1)}{K(1-\frac{|S|}{2^{2n}})},
			\end{equation}
			where
			\begin{equation}
			\frac{1}{K}=1-\cos(\lceil\frac{|g|}{2}\rceil\frac{2}{|g|}\pi).
			\end{equation}
		\end{restatable}
		\begin{proof}
			First, if $g=e$, the verifier can immediately reject the proof since $e$ is contained in every group. In this case $\Pr(\prove(g,\rho_s)=1)=0$ and the inequality holds. In this following proof, we assume $g\neq e$ and therefore $|g|\neq 1$.\\
			
			To prove \autoref{equ-prove-pass}, we just need to prove it for any pure state $\rho_r=\ket{\psi}\bra{\psi}$, because if this theorem is true for any pure state, then for any mix state
			$
			\rho=\sum_i q_i \ket{i}\bra{i},
			$
			\begin{equation}
			\begin{split}
			\Pr(\core(g,\rho_r)=1)& =\sum_i q_i \Pr(\core(g,\ket{i})=1)\\
			&\leq \sum_i q_i \frac{1-\Pr(T(\ket{i})=1)}{K(1-\frac{|S|}{2^{2n}})}\\
			&=\frac{1-\Pr(T(\rho_r)=1)}{K(1-\frac{|S|}{2^{2n}})}.
			\end{split}
			\end{equation}
			Therefore, we start to prove this theorem for pure state here. We can extend any pure state $\ket{\psi}$ as
			\begin{equation}
			\ket{\psi}=\sum_{\alpha\in G} \beta_\alpha \ket{\alpha}
			\end{equation}
			For simplicity, in this proof we denote $\Pr(\core(g,\ket{\psi})=1)$ by $p(g)$. We want to prove
			\begin{equation}
			p(g)\leq \frac{1-\Pr(T(\ket{\psi})=1)}{K(1-\frac{|S|}{2^{2n}})}.
			\end{equation}
			In \autoref{alg-qma_veri}, the state is tested as
			\begin{equation}
			\begin{split}
		&\sum_{\alpha\in G} \beta_\alpha \ket{\alpha}\rightarrow \frac{1}{2} \ket{0} \sum_{\alpha\in G}  \beta_\alpha (\ket{\alpha}+\ket{\alpha g}) +
			\frac{1}{2} \ket{1} \sum_{\alpha\in G}  \beta_\alpha (\ket{\alpha}-\ket{\alpha g}).
			\end{split}
			\end{equation}
			The probability for a element $g\in G$ to be verified by the state is
			\begin{equation}
			\begin{split}
			p(g)&= \frac{1}{4} \left\Vert \sum_{\alpha\in G}  \beta_\alpha (\ket{\alpha}-\ket{\alpha g})\right\Vert^2\\
			&= \frac{1}{4}  \sum_{\alpha\in G}( |\beta_\alpha|^2 + |\beta_{\alpha g^{-1}}|^2 -2Re(\beta_\alpha^* \beta_{\alpha g^{-1}}))\\
			&= \frac{1}{4}  ( 2 - \sum_{\alpha\in G} 2Re(\beta_\alpha^* \beta_{\alpha g^{-1}})),\\
			\end{split}
			\end{equation}
			by which we can obtain
			\begin{equation}
			\label{equ-repp-p}
			\sum_{\alpha\in G} Re(\beta_\alpha^* \beta_{\alpha g^{-1}})=1-2p(g).
			\end{equation}
			Decompose $G$ by the cosets $G/S=\{\alpha S|\alpha\in G\}$, where $\alpha S=\{\alpha s|s\in S\}$. Easily one can find that if $h\in \alpha S$, $hS=\alpha S$. By this decomposition we sum up \autoref{equ-repp-p} and obtain
			\begin{equation}
			\begin{split}
			\sum_{s \in S} (1-2p(s))
			&=\sum_{s\in S}\sum_{\alpha\in G} Re(\beta_\alpha^* \beta_{\alpha s^{-1}})\\
			&=Re(\sum_{\alpha\in G} \beta_\alpha^* \sum_{s\in \alpha S}\beta_{s})\\
			&=Re(\sum_{\alpha S\in G/S}\sum_{h\in \alpha S}\beta_h^*\sum_{s\in h S}\beta_{s})\\
			&=Re(\sum_{\alpha S\in G/S}\sum_{h\in \alpha S}\beta_h^*\sum_{s\in \alpha S}\beta_{s})\\
			&=\sum_{\alpha S\in G/S} |\sum_{h\in \alpha S}\beta_h|^2,
			\end{split}
			\end{equation}
			and therefore
			\begin{equation}
			\sum_{s\in S}p(s)=\frac{1}{2}(|S|-\sum_{\alpha S\in G/S} |\sum_{h\in \alpha S}\beta_h|^2).
			\end{equation}
			
			On the other hand, we want to get the probability $\Pr(T(\ket{\psi})=1)$ for this state to pass the test. Notice that the probability for the state to pass the test by $s$ is
			\begin{equation}
			\Pr (\core(s,\ket{\psi})=0)=1-\Pr (\core(s,\ket{\psi})=1)=1-p(s).
			\end{equation}
			Also, in a test the subgroup element used is randomly sampled by Babai's algorithm. By \autoref{thm-babai-elementsampler} we know that in polynomial time we can sample every subgroup element $s$ with
			\begin{equation}
			\Pr(s\text{ is sampled})\in (\frac{1}{|S|}-\frac{1}{2^{2n}}, \frac{1}{|S|}+\frac{1}{2^{2n}}).
			\end{equation}
			Therefore we have
			\begin{equation}
			\begin{split}
			\Pr(T(\ket{\psi})=1)&=\sum_{s\in S} \Pr(s\text{ is sampled}) (1-p(s))\\
			&=1-\sum_{s\in S} \Pr(s\text{ is sampled}) p(s)\\
			&\leq 1-(\frac{1}{|S|}-\frac{1}{2^{2n}})\frac{1}{2}(|S|-\sum_{\alpha S\in G/S} |\sum_{h\in \alpha S}\beta_h|^2)\\
			&= \frac{1}{2}+\frac{|S|}{2^{2n+1}}+\frac{1}{2}(\frac{1}{|S|}-\frac{1}{2^{2n}})\sum_{\alpha S\in G/S} |\sum_{h\in \alpha S}\beta_h|^2.\\
			\end{split}
			\end{equation}
			To give a relation between $p(g)$ and $\Pr(T(\ket{\psi})=1)$, we want to know the maximum of
			\begin{equation}
			f(\vec{\beta})=\sum_{\alpha S\in G/S} |\sum_{h\in \alpha S}\beta_h|^2
			\end{equation}
			under the condition that
			\begin{equation}
			b(\vec{\beta})=\sum_{\alpha\in G} Re(\beta_\alpha^* \beta_{\alpha g^{-1}})=B=1-2p(g)
			\end{equation}
			and
			\begin{equation}
			l(\vec{\beta})=\sum_{h\in G} |\beta_h|^2=1.
			\end{equation}
			We should survey more carefully the structure of the coset $\alpha S$. Define $cycle(c,g)$ as $c\langle g \rangle = \{cg^i|i\in Z\}$. A coset can be decomposed into disjoint orbits. Let $CYC(\alpha S,g)$ be a set of $c$ such that for any $c_1,c_2\in CYC(\alpha S,g)$, $cycle(c_1,g)\neq cycle(c_2,g)$ if $c_1\neq c_2$; and $\alpha S=\cup_{c\in CYC(\alpha S,g)} cycle(c,g)$.
			By these definition we can further decompose the summation as
			\begin{equation}
			\sum_{\alpha\in G} Re(\beta_\alpha^* \beta_{\alpha g^{-1}})=\sum_{\alpha S\in G/S} \sum_{c\in CYC(\alpha S,g)} \sum_{k=1}^{|g|} Re(\beta_{cg^{-k+1}}^*\beta_{cg^{-k}}).
			\end{equation}
			Define
			\begin{equation}
			f_\alpha(\vec{\beta})=\sum_{c\in CYC(\alpha S,g)} \sum_{k=1}^{|g|} \beta_{cg^{-k}},
			\end{equation}
			\begin{equation}
			b_\alpha(\vec{\beta})=\sum_{c\in CYC(\alpha S,g)} \sum_{k=1}^{|g|} Re(\beta_{cg^{-k+1}}^*\beta_{cg^{-k}}),
			\end{equation}
			\begin{equation}
			l_\alpha(\vec{\beta})=\sum_{c\in CYC(\alpha S,g)} \sum_{k=1}^{|g|} |\beta_{cg^{-k}}|^2.
			\end{equation}
			Therefore
			\begin{equation}
			f(\vec{\beta})=\sum_{\alpha S\in G/S} |\sum_{c\in CYC(\alpha S,g)} \sum_{k=1}^{|g|} \beta_{cg^{-k}} |^2=\sum_{\alpha S\in G/S} |f_\alpha(\vec{\beta})|^2.
			\end{equation}
			We want to know the maximum of $f_\alpha$ when $b_\alpha$ and $l_\alpha$ are fixed. To find out this relationship we first define
			\begin{equation}
			o_{c}(\vec{\beta})=|\sum_{k=1}^{|g|} \beta_{cg^{-k}}|^2,
			\end{equation}
			\begin{equation}
			b_{c}(\vec{\beta})=\sum_{k=1}^{|g|} Re(\beta_{cg^{-k+1}}^*\beta_{cg^{-k}}),
			\end{equation}
			\begin{equation}
			l_{c}(\vec{\beta})=\sum_{k=1}^{|g|} |\beta_{cg^{-k}}|^2.
			\end{equation}
			We need to study the maximum of $o_{c}$ when $g_{c}$ and $l_{c}$ are fixed. Decompose every $\beta_h$ into real part and imaginary part as
			\begin{equation}
			\beta_h=R_h+iI_h.
			\end{equation}
			Therefore
			\begin{equation}
			l_{c}(\vec{\beta})=\sum_{n=1}^{|g|} |\beta_{cg^{-n}}|^2=\sum_{n=1}^{|g|} R_{cg^{-n}}^2 + \sum_{n=1}^{|g|} I_{cg^{-n}}^2=l_{c}^R(\vec{R})+l_{c}^I(\vec{I}),
			\end{equation}
			\begin{equation}
			b_{c}(\vec{\beta})=\sum_{n=1}^{|g|} Re(\beta_{cg^{-n+1}}^*\beta_{cg^{-n}}) = \sum_{n=1}^{|g|} R_{cg^{-n+1}}^*R_{cg^{-n}} + \sum_{n=1}^{|g|} I_{cg^{-n+1}}^*I_{cg^{-n},}=b_{c}^R(\vec{R})+b_{c}^I(\vec{I}),
			\end{equation}
			\begin{equation}
			o_{c}(\vec{\beta})=|\sum_{n=1}^{|g|} \beta_{cg^{-n}}|^2=|\sum_{n=1}^{|g|} R_{cg^{-n}}|^2+|\sum_{n=1}^{|g|} I_{cg^{-n}}|^2=o_{c}^R(\vec{R})+o_{c}^I(\vec{I}).
			\end{equation}
			By \autoref{thm-maximumo} we know that when $b_c^R$ and $l_c^R$ are fixed,
			\begin{equation}
			\max o_{c}^R(\vec{R})=K|g|(b_{c}^R-l_{c}^R)+|g|l_{c}^R=K|g|(b_{c}^R+\frac{1-K}{K}l_{c}^R),
			\end{equation}
			where
			\begin{equation}
			K^{-1}=1-\cos(\lceil\frac{|g|}{2}\rceil\frac{2}{|g|}\pi).
			\end{equation}
			Then we get $\max o_{c}(\vec{\beta})$ under the condition that $b_{c}^R(\vec{R})+b_{c}^I(\vec{I})=b_{c}(\vec{\beta})$ and $l_{c}^R(\vec{R})+l_{c}^I(\vec{I})=l_{c}(\vec{\beta})$. That is
			\begin{equation}
			\begin{split}
			\max o_{c}(\vec{\beta})
			=& \max(o_{c}^R(\vec{R})+o_{c}^I(\vec{I}))\\
			=&\max \left( K|g|(b_{c}^R+\frac{1-K}{K}l_{c}^R)+K|g|(b_{c}^I+\frac{1-K}{K}l_{c}^I)\right)\\
			=&\max \left( K|g|(b_{c}^R+b_{c}^I+\frac{1-K}{K}(l_{c}^R+l_{c}^I))\right)\\
			=&K|g|(b_{c}+\frac{1-K}{K}l_{c}).
			\end{split}
			\end{equation}
			Therefore we have
			\begin{equation}
			\begin{split}
			\max f_\alpha(\vec{\beta})&\leq \sum_{c\in CYC(\alpha S,g)} \sqrt{\max o_{c}}\\
			&=\frac{|S|}{|g|} \sum_{c\in CYC(\alpha S,g)} \frac{|g|}{|S|} \sqrt{\max o_{c}}\\
			&\leq \frac{|S|}{|g|}  \sqrt{ \sum_{c\in CYC(\alpha S,g)} \frac{|g|}{|S|} \max o_{c}}\\
			&= \sqrt{\frac{|S|}{|g|}}  \sqrt{ \sum_{c\in CYC(\alpha S,g)}K|g|(b_{c}+\frac{1-K}{K}l_{c}) }\\
			&= \sqrt{K|S| }  \sqrt{ \sum_{c\in CYC(\alpha S,g)}(b_{c}+\frac{1-K}{K}l_{c}) }\\
			&= \sqrt{K|S|}  \sqrt{ b_{\alpha}+\frac{1-K}{K}l_{\alpha} }.\\
			\end{split}
			\end{equation}
			Then
			\begin{equation}
			\begin{split}
			\max f(\vec{\beta})&=\max \sum_{\alpha S\in G/S} |f_\alpha(\vec{\beta})|^2\\
			&\leq \sum_{\alpha S\in G/S} |\max f_\alpha(\vec{\beta})|^2\\
			&=\sum_{\alpha S\in G/S} K|S| ( b_{\alpha}+\frac{1-K}{K}l_{\alpha}) \\
			&=K|S|(B+\frac{1-K}{K})=|S|(KB+1-K).
			\end{split}
			\end{equation}
			Finally
			\begin{equation}
			\begin{split}
			\Pr(T(\ket{\psi})=1)&\leq \frac{1}{2}+\frac{|S|}{2^{2n+1}}+(\frac{1}{|S|}-\frac{1}{2^{2n}})\frac{1}{2}|S|(KB+1-K)\\
			&=\frac{1}{2}+\frac{|S|}{2^{2n+1}}+(\frac{1}{|S|}-\frac{1}{2^{2n}})\frac{1}{2}|S|(K(1-2p(g))+1-K)\\
			&=\frac{1}{2}+\frac{|S|}{2^{2n+1}}+(\frac{1}{|S|}-\frac{1}{2^{2n}})\frac{|S|}{2}(1-2Kp(g))\\
			&=\frac{1}{2}+\frac{|S|}{2^{2n+1}}+\frac{1}{2}(1-2Kp(g))-\frac{|S|}{2^{2n+1}}(1-2Kp(g))\\
			&=1+(\frac{|S|}{2^{2n}}-1)Kp(g).\\
			\end{split}
			\end{equation}
			\begin{equation}
			\begin{split}
			p(g)&\leq \frac{1-\Pr(T(\ket{\psi})=1)}{K(1-\frac{|S|}{2^{2n}})}.
			\end{split}
			\end{equation}

		\end{proof}
		
		 Here we show how to deduce the claim in the main text by \autoref{thm-pass-soundness}.
		 \begin{proof}
		 Notice that
		$
		    \frac{1}{K}=1-\cos(\lceil\frac{|g|}{2}\rceil\frac{2}{|g|}\pi)\leq 2
		$
		,
		$
		    |S|\leq 2^n
		$
		and
		$
		    (1-\frac{|S|}{2^{2n}})\geq \frac{1}{2},
		$
		by the proposition below we can have
		\begin{equation}
		    	\Pr(\core(g,\rho_r)=1)\leq 4 (1-\Pr(T(\rho_r)=1)).
		\end{equation}
		 \end{proof}

		\begin{proposition}
			\label{thm-omax-1}
			Let $R$ be a vector of  $n$ real numbers $(n>1)$. Under the condition that  $\sum_{i=1}^{n} R_i^2=1$ and
			$
			\sum_{i=1}^{n} R_iR_{i+1}=b
			$
			($R_{n+1}$ is defined to be $R_1$), the maximum of
			$
			O=(\sum_{i=1}^n R_i)^2
			$
			is
			$$
			O=\frac{n}{1-\cos(\lceil\frac{n}{2}\rceil\frac{2}{n}\pi)}(b-1)+n.
			$$
		\end{proposition}
		
		\begin{proof}
			The Lagrangian multiplier of this problem is
			\begin{equation}
			\begin{split}
			F(R,\lambda)&=(\sum_{i=1}^n R_i)^2+\lambda_1 (\sum_{i=1}^{n} R_iR_{i+1}-b)+\lambda_2(\sum_{i=1}^n R_i^2-1).\\
			\end{split}
			\end{equation}
			The derivatives should be zero when $O$ is maximized
			\begin{equation}
			\label{equ-derivative}
			\frac{\partial F}{\partial R_k}=2(\sum_{i=1}^n R_i)+\lambda_1(R_{k-1}+R_{k+1})+2\lambda_2R_k=0.
			\end{equation}
			Sum up the derivatives, we have
			\begin{equation}
			\sum_{k=1}^n\frac{\partial F}{\partial R_k}=2n(\sum_{i=1}^n R_i)+2\lambda_1(\sum_{i=1}^n R_i)+2\lambda_2(\sum_{i=1}^n R_i)=0
			\end{equation}
			and therefore
			\begin{equation}
			\label{equ-sum-lambda}
			\lambda_1+\lambda_2=-n.
			\end{equation}
			In a similar way, we have
			\begin{equation}
			\frac{\partial F}{\partial R_k}R_k=2(\sum_{i=1}^n R_i)R_k+\lambda_1(R_{k-1}R_k+R_kR_{k+1})+2\lambda_2R_k^2=0
			\end{equation}
			and therefore
			\begin{equation}
			\sum_{k=1}^n \frac{\partial F}{\partial R_k}R_k=2O+2\lambda_1b+2\lambda_2=0.
			\end{equation}
			\begin{equation}
			\label{equ-o-lambda}
			O=-\lambda_1 b+\lambda_1+n.
			\end{equation}
			Also, by \autoref{equ-derivative}, we have
			\begin{equation}
			\lambda_1(R_{k-1}+R_{k+1})+2\lambda_2R_k=\lambda_1(R_{k}+R_{k+2})+2\lambda_2R_{k+1},
			\end{equation}
			and therefore
			\begin{equation}
			R_{k}+(2\frac{\lambda_2}{\lambda_1}-1)R_{k-1}-(2\frac{\lambda_2}{\lambda_1}-1)R_{k-2}-R_{k-3}=0.
			\end{equation}
			The formula for each $R_k$ can be obtained by the characteristic equation
			\begin{equation}
			x^3+\Lambda x^2 - \Lambda x-1=0,
			\end{equation}
			where
			\begin{equation}
			\Lambda=(2\frac{\lambda_2}{\lambda_1}-1).
			\end{equation}
			The solutions are
			\begin{equation}
			\label{equ-solu-char}
			x=1,\pm\sqrt{(\frac{\lambda_2}{\lambda_1})^2-1}-\frac{\lambda_2}{\lambda_1}.
			\end{equation}
			Denote the latter two solutions as $X_+$ and $X_-$. Easily, one can verify that we can define $X$ as $X_+=X_-^{-1}=X$. We first point out the relation of $X$ and $\lambda_1$ here. By \autoref{equ-solu-char}, we have
			\begin{equation}
			\label{equ-lambda1X}
			\lambda_1=\frac{2 n X}{(X-1)^2}.
			\end{equation}
			Then, combining with \autoref{equ-o-lambda}, we obtain
			\begin{equation}
			\label{equ-O-X}
			O=\lambda_1 (1-b)+n=\frac{2 n X}{(X-1)^2}(1-b)+n.
			\end{equation}
			\textbf{When $n$ is even:}\\
			First, easily one can prove that $O=0$ when $b=-1$, which is denoted by $O(-1)=0$. When $b=-1$, if any $|R_i|$ is not equal to $|R_{i+1}|$, by rearrangement inequity we have
			\begin{equation}
			\sum_i |R_i| |R_{i+1}|<\sum_i |R_i| |R_i|=1
			\end{equation}
			and
			\begin{equation}
			b=\sum_i R_iR_{i+1}>-\sum_i |R_i| |R_{i+1}|=-1,
			\end{equation}
			which is contradictory to our assumption that $b=-1$. Therefore, we must have $|R_i|=|R_{i+1}|$ and easily one can find that $R_i=-R_{i+1}$ to make $b=-1$. Then, $O(-1)=(\sum_i R_i)^2=0$ when $n$ is even. Finally we have
			\begin{equation}
			O(-1)=0=\frac{4 n X}{(X-1)^2}+n.
			\end{equation}
			Solving this we have $X=-1$, and then $\lambda_1 = - n/2$. Therefore,
			\begin{equation}
			O=\frac{n}{2}(b-1)+n=\frac{n}{1-\cos(\pi)}(b-1)+n=\frac{n}{1-\cos(\lceil\frac{n}{2}\rceil\frac{2}{n}\pi)}(b-1)+n.
			\end{equation}
			\begin{comment}
			By solve this equation, we know
			\begin{equation}
			X=-1,X_+=X_-.
			\end{equation}
			The formula for $R_k$ should be
			\begin{equation}
			\label{equ-Rk-even}
			R_k=A_1+(A_2+A_3k)(-1)^k.
			\end{equation}
			One can plug \autoref{equ-Rk-even} in \autoref{equ-derivative} we can derive
			\begin{equation}
			O=(\sum_{k=1}^n R_k)^2=(nA_1)^2,
			\end{equation}
			but one can also get
			\begin{equation}
			O=(\sum_{k=1}^n R_k)^2=\sum_{k=1}^n(A_1+(A_2+A_3k)(-1)^k)=(nA_1+A_3\frac{n}{2})^2.
			\end{equation}
			Then one can conclude that
			\begin{equation}
			A_3=0,
			\end{equation}
			because $A_3\frac{n}{2}=-2nA_1$ is impossible as $||$
			Then
			\begin{equation}
			R_k=A_1+A_2(-1)^k,
			\end{equation}
			which is equivalent to when we assume $X_+\neq X_-$.\\
			\end{comment}
			\textbf{When $n$ is odd:}\\
			We first assume $X_+\neq X_-$ and therefore $X\neq \pm 1$. The formula for $R_k$ in this case should be
			\begin{equation}
			\label{equ-Rk}
			R_k=A_1+A_2X^k+A_3X^{-k},
			\end{equation}
			where $A_1,A_2$ and $A_3$ are parameters that should be determined. Plug \autoref{equ-Rk} into \autoref{equ-derivative} and notice that $X_++X_-=-2\frac{\lambda_2}{\lambda_1}$, we have
			\begin{equation}
			\begin{split}
			0=&2\sum_{k=1}^n R_k+\lambda_1(A_1+A_2X^{k-1}+A_3X^{-(k-1)}+A_1+A_2X^{k+1}+A_3X^{-(k+1)})+2\lambda_2R_k\\
			=&2\sum_{k=1}^n R_k+\lambda_1(2A_1+A_2X^k(X_-+X_+)+A_3X^{-k}(X_-+X_+))+2\lambda_2R_k\\
			=&2\sum_{k=1}^n R_k+2(\lambda_1+\lambda_2)A_1-2\lambda_2(A_2X^k+A_3X^{-k})+2\lambda_2(A_2X^k+A_3X^{-k})\\
			=&2\sum_{k=1}^n R_k+2(\lambda_1+\lambda_2)A_1 \ .
			\end{split}
			\end{equation}
			%It's too long so I hide it
			\begin{comment}
			\begin{equation}
			2\sum_{k=1}^n R_k+\lambda_1(2A_1+A_2X^k(X^{-1}+X^{+1})+A_3X^{-k}(X^{-1}+X^{+1}))+2\lambda_2R_k=0
			\end{equation}
			
			\begin{equation}
			2\sum_{k=1}^n R_k+\lambda_1(2A_1+A_2X^k(X_-+X_+)+A_3X^{-k}(X_-+X_+))+2\lambda_2R_k=0
			\end{equation}
			
			\begin{equation}
			2\sqrt{-\lambda_1 b+\lambda_1+n}+\lambda_1(2A_1+A_2X^k(-2\frac{\lambda_2}{\lambda_1})+A_3X^{-k}(-2\frac{\lambda_2}{\lambda_1}))+2\lambda_2R_k=0
			\end{equation}
			
			\begin{equation}
			\sqrt{-\lambda_1 b+\lambda_1+n}+\lambda_1(A_1+A_2X^k(-\frac{\lambda_2}{\lambda_1})+A_3X^{-k}(-\frac{\lambda_2}{\lambda_1}))+\lambda_2(A_1+A_2X^k+A_3X^{-k})=0
			\end{equation}
			
			\begin{equation}
			\sqrt{-\lambda_1 b+\lambda_1+n}+(\lambda_1+\lambda_2)A_1-\lambda_2A_2X^k-\lambda_2A_3X^{-k}+\lambda_2(A_2X^k+A_3X^{-k})=0
			\end{equation}
			
			\begin{equation}
			\sqrt{-\lambda_1 b+\lambda_1+n}+(\lambda_1+\lambda_2)A_1=0
			\end{equation}
			\end{comment}
			Additionally, because of \autoref{equ-sum-lambda}, we have
			\begin{equation}
			\sum_{k=1}^n R_k=nA_1.
			\end{equation}
			In the meantime we can just sum up $R_k$ to get $\sum_{k=1}^n R_k$ as
			\begin{equation}
			\sum_{k=1}^n R_k=nA_1=\sum_{k=1}^n (A_1+A_2X^k+A_3X^{-k}).
			\end{equation}
			Therefore,
			\begin{equation}
			\label{equ-sum2zero}
			\sum_{k=1}^n (A_2X^k+A_3X^{-k})=0,
			\end{equation}
			\begin{equation}
			A_2\frac{X(1-X^n)}{1-X}+A_3\frac{X^{-1}(1-X^{-n})}{1-X^{-1}}=0,
			\end{equation}
			\begin{equation}
			\label{equ-A3A2}
			A_3=A_2\frac{X-X^{n+1}}{1-X^{-n}}=-A_2X^{n+1}.
			\end{equation}
			On the other hand, we have
			\begin{equation}
			\sum_{k=1}^n R_kR_{k+1}=\sum_{k=1}^{n-1} R_kR_{k+1} + R_n R_1 =  A_2^2L_1(X)+A_1^2 n=b  \ ,
			\end{equation}
			where
			\begin{equation}
			L_1(X)=\frac{\left(-X^2+2 X+1\right) X^{2 n}+\left(-n X^4+n+X^4+2 X^3-2 X-1\right) X^n-X^4-2 X^3+X^2}{X^2-1}.
			\end{equation}
			Therefore,
			\begin{equation}
			A_2^2=\frac{b-A_1^2 n}{L_1(X)}.
			\end{equation}
			Plug $A_2^2$ in the the summation of $R_k^2$
			\begin{equation}
			\sum_{k=1}^n R_k^2=A_1^2 n-\frac{2 A_2^2 X \left(-X^{2 n+1}+n \left(X^2-1\right) X^n+X\right)}{X^2-1}=1,
			\end{equation}
			then we have
			\begin{equation}
			\label{equ-A1X}
			A_1^2 n-\frac{2 (b-A_1^2 n) X \left(-X^{2 n+1}+n \left(X^2-1\right) X^n+X\right)}{L_1(X) (X^2-1)}=1.
			\end{equation}
			\begin{comment}
			\begin{equation}
			A_1^2 n-\frac{2 X \left(-X^{2 n+1}+n \left(X^2-1\right) X^n+X\right) \left(b-A_1^2 n\right)}{(1-(X-2) X) X^{2 n}-\left(X^2-1\right) \left(n \left(X^2+1\right)-(X+1)^2\right) X^n-(X+2) X^3+X^2}=1
			\end{equation}
			\begin{equation}
			A_1^2=\frac{(X (2 b X+X-2)-1) X^{2 n}+\left(X^2-1\right) X^n \left(n \left(-2 b X+X^2+1\right)-(X+1)^2\right)+X^2 (-2 b+X (X+2)-1)}{n (X-1) \left((3 X+1) X^{2 n}+(X+1) \left(n (X-1)^2-(X+1)^2\right) X^n+(X+3) X^2\right)}
			\end{equation}
			\begin{equation}
			\frac{2 n X ((b-1) n+2) \left(-X^{2 n+1}+n \left(X^2-1\right) X^n+X\right)}{(b-1) \left(X^2
			\left(4-n (X-1)^2\right)+\left(n (X-1)^2-4 X^2\right) X^{2 n}+n \left(X^2-1\right)
			\left(n (X-1)^2-(X-6) X-1\right) X^n\right)}
			\end{equation}
			\end{comment}
			Also, notice that $O=(\sum_{k=1}^n R_k)^2 =  n^2 A_1^2$ and \autoref{equ-O-X}. We have
			\begin{equation}
			\label{equ-A1X-2}
			n^2 A_1^2=\frac{2 n X}{(X-1)^2}(1-b)+n.
			\end{equation}
			Combine \autoref{equ-A1X} and \autoref{equ-A1X-2}, we have
			\begin{equation}
			\begin{split}
			0&=-\frac{2 b n X}{(X-1)^2}+\frac{2 n X}{(X-1)^2}+n-n^2 A_1^2\\
			&=\frac{2 (b-1) n X (X+1)^2 \left(X-X^n\right) \left(X^n-1\right)}{(X-1)^2 \left((3 X+1) X^{2 n}+(X+1) \left(n (X-1)^2-(X+1)^2\right) X^n+(X+3) X^2\right)}.
			\end{split}
			\end{equation}
			A necessary condition for the above equation is
			\begin{equation}
			2 (1-b) n X^2 (X+1)^2 \left(X^{n-1}-1\right) \left(X^n-1\right)=0,
			\end{equation}
			or equivalently,
			\begin{equation}
			X=0,X=-1,X=e^{i\frac{m}{n}2\pi},X=e^{i\frac{m}{n-1}2\pi}.
			\end{equation}
			We needs to find out which solution of $X$ to adopt. Define $t$ as $X=e^{it}$. Then we have
			\begin{equation}
			\frac{2 X}{(X-1)^2}=\frac{2 e^{i t}}{\left(-1+e^{i t}\right)^2}=\frac{1}{\cos (t)-1}
			\end{equation}
			and
			\begin{equation}
			O=\frac{n}{1-\cos(t)}(b-1)+n.
			\end{equation}
			Obviously $X\neq 0$ because $X\neq -1$, $t\neq \pi$. $O(b)$ is maximized when $1-\cos(t)$ is maximized and $t \mod 2\pi$ should be as closed to $\pi$ as possible. Here, we assume $n>3$ first. For the case $n=3$ it is trivial to see the result is the same. If $X=e^{i\frac{m}{n-1}2\pi}$, $O$ is maximized when $t=\pi\pm\frac{2\pi}{n-1}$. If $X=e^{i\frac{m}{n}2\pi}$, $O$ is maximized when $t=\pi\pm\frac{\pi}{n}$, which is closer to $\pi$ than $t=\pi\pm\frac{2\pi}{n-1}$. Therefore $O$ is maximized when
			\begin{equation}
			t=\pi\pm\frac{\pi}{n}=\frac{n+1}{n}\pi.
			\end{equation}
			Therefore,
			\begin{equation}
			O(b)=\frac{n}{1-\cos(\frac{n+1}{n}\pi)}(b-1)+n=\frac{n}{1-\cos(\lceil\frac{n}{2}\rceil\frac{2}{n}\pi)}(b-1)+n.
			\end{equation}
			Finally, we must deal with the condition when $X_+=X_-=X=-1$. We will prove $X\neq -1$. Notice that, in this case, when $b=-1$, by \autoref{equ-O-X}, $O$ must be zero.
			\begin{equation}
			\label{equ-Rk-degenerate}
			R_k=A_1+(A_2+A_3k)(-1)^k.
			\end{equation}
			Because $O=0$,
			\begin{equation}
			\sum_{k=1}^n R_k=\frac{1}{2} \left(2 A_1 n- A_3 n-2 A_2- A_3\right)=0,
			\end{equation}
			\begin{equation}
			A_1=\frac{A_3 n+2 A_2+A_3}{2 n}.
			\end{equation}
			Calculate $\sum_{k=1}^n R_k^2$ and $\sum_{k=1}^n R_kR_{k+1}$ and substitute $A_1$ by above equation, we have
			\begin{equation}
			\begin{split}
			\sum_{k=1}^n R_k^2=1=&A_1^2 n+A_2^2 n+\frac{1}{6} A_3^2 n (n+1) (2 n+1)+A_3 (n+1) \left(A_2 n-A_1\right)-2 A_2 A_1\\
			=&\frac{1}{12 n}\left(n^2-1\right) \left(A_3^2 \left(4 n^2+6 n+3\right)+12 A_3 A_2 (n+1)+12 A_2^2\right),
			\end{split}
			\label{eq:A2}
			\end{equation}
			\begin{equation}
			\begin{split}
			\sum_{k=1}^n R_kR_{k+1}=-1=&A_1^2 n-A_1 \left(A_3 (n+1)+2 A_2\right)-\frac{1}{3} (n-2) \left(3 A_3 A_2 (n+1)+A_3^2 n (n+2)+3 A_2^2\right)\\
			=&-\frac{1}{12 n}(n-1) \left(12 A_3 A_2 \left(n^2-1\right)+12 A_2^2 (n-1)+A_3^2 (n (4 n (n+1)-9)-3)\right).
			\end{split}
			\end{equation}
			By the above two equations we can get
			\begin{equation}
			\frac{(n-1) \left(A_3^2 (-(n-3)) (n+1)-6\right)}{6 (n+1)}=-1,
			\end{equation}
			which has no solution when $n=3$. Thus, $X\neq -1$ when $n=3$. Assume $n>3$, we have
			\begin{equation}
			A_3^2=\frac{12}{(n-3) (n-1) (n+1)} .
			\end{equation}
			Viewing \autoref{eq:A2} as a quadratic equation for $A_2$, the discriminant is
			\begin{equation}
			\Delta = \frac{48 n (12 + A_3^2(n-n^3))}{n^2-1} = \frac{1728 n}{(3-n)(n^2-1)}.
			\end{equation}
			Therefore, for $n>3$, $A_2$ does not have real solution. Thus, we proved that $X\neq -1$ for every odd $n$ ($n>1$).
			\\
			\textbf{In conclusion, for both even and odd $n$}:
			\begin{equation}
			O=\frac{n}{1-\cos(\lceil\frac{n}{2}\rceil\frac{2}{n}\pi)}(b-1)+n.
			\end{equation}
			
		\end{proof}

		\begin{restatable}{corollary}{maximumo}
			\label{thm-maximumo}
			Let $R$ be a vector of  $n$ real numbers. Under the condition that  $\sum_{i=1}^{n} R_i^2=l\,(l\in(0,1])$ and
			$
			\sum_{i=1}^{n} R_iR_{i+1}=b
			$
			($R_{n+1}$ is defined to be $R_1$), the maximum of
			$
			O=(\sum_{i=1}^n R_i)^2
			$
			is
			$$
			O=\frac{n}{1-\cos(\lceil\frac{n}{2}\rceil\frac{2}{n}\pi)}(b-l)+nl.
			$$
		\end{restatable}
		\begin{proof}
			Assume the $R$ that maximized $O$ under the condition that $\sum_{i=1}^nR_i^2=l$ and $\sum_{i=1}^nR_iR_{i+1}=b$. Then the vector $T=\frac{R}{\sqrt{l}}$ satisfies the condition that $\sum_{i=1}^nT_i^2=1$ and $\sum_{i=1}^nT_iT_{i+1}=\frac{b}{l}$. By \autoref{thm-omax-1} we know that
			\begin{equation}
			\label{equ-maximumo1}
			(\sum_{i=1}^nT_i)^2=\frac{1}{l}(\sum_{i=1}^nR_i)^2\leq \frac{n}{1-\cos(\lceil\frac{n}{2}\rceil\frac{2}{n}\pi)}(\frac{b}{l}-1)+n.
			\end{equation}
			If the theorem we want to prove is not true and
			\begin{equation}
			(\sum_{i=1}^nR_i)^2>\frac{n}{1-\cos(\lceil\frac{n}{2}\rceil\frac{2}{n}\pi)}(b-l)+nl,
			\end{equation}
			then
			\begin{equation}
			\label{equ-maximumo2}
			(\sum_{i=1}^nT_i)^2=\frac{1}{l}(\sum_{i=1}^nR_i)^2>\frac{n}{1-\cos(\lceil\frac{n}{2}\rceil\frac{2}{n}\pi)}(\frac{b}{l}-1)+n.
			\end{equation}
			\autoref{equ-maximumo1} and \autoref{equ-maximumo2} are contradictory and then we prove the theorem.
		\end{proof}
		
		\begin{theorem}[\textbf{Babai}~\citep{Babai91}]
			\label{thm-babai-elementsampler}
			For any group oracle B there exists a randomized process $\mathcal{P}$ acting as follows. On input $g_1,\dots,g_k\in G(B_n)$ and $\epsilon > 0$, $\mathcal{P}$ outputs an element of $H = \langle g_1,\dots, g_k \rangle$ in time polynomial in $n + \log 1/\epsilon$ such that each $g\in H$ is output with probability in the range $(1/|H|-\epsilon, 1/|H| + \epsilon)$.
		\end{theorem}
		
\subsection{More experimental results}

For the input state $\frac{1}{\sqrt{2}}(\ket{\mathsf{B}}+\ket{\mathsf{AB}})=\frac{1}{\sqrt{2}}(\ket{HH}+\ket{VV})$, we present the detailed experimental imaginary matrix of the outputs here, as shown in Fig. \ref{figs:Expim}.

\begin{figure*}[htbp]
\centering
\includegraphics[clip=true, width=\linewidth]{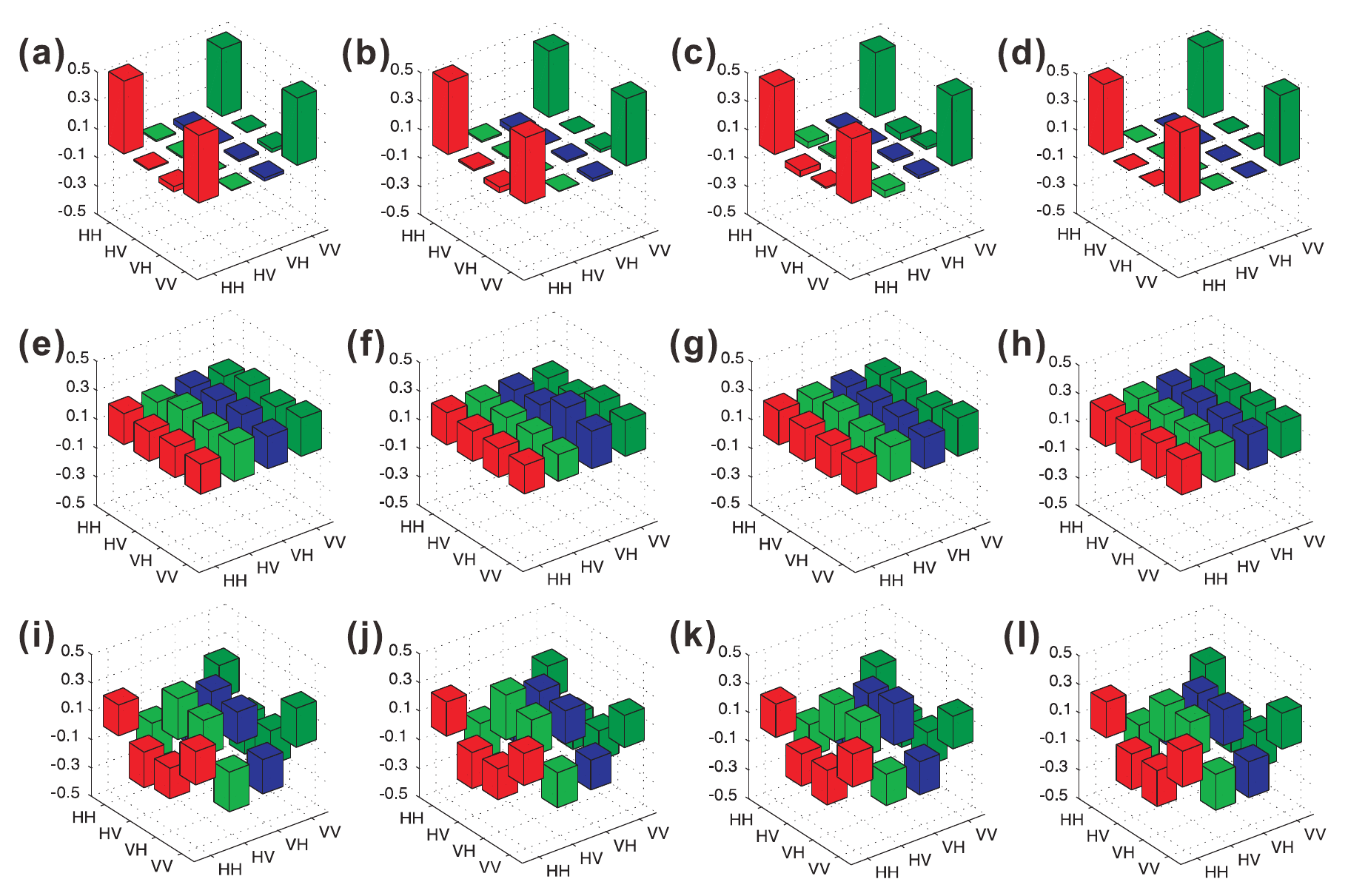} %
\caption{(Color online). Real parts of density matrices of the finial output photons for the case with input state of $(\ket{HH}+\ket{VV})/\sqrt{2}$. {\bf a-c} represent the cases of initial state, output photons of $a3$ and $b$ in $\E$-type interferometer, output photons of $a$ and $b3$ in $\mathsf{AB}$-type interferometer without CNOT gate, respectively. {\bf e-g} represent the cases of output photons of $a3$ and $b3$ in $\A$-type interferometer (with fidelity $92.6\pm2.4\%$), $a3$ and $b$ in $\B$-type interferometer ($88.9\pm0.7\%$), $a$ and $b3$ in $\mathsf{AB}$-type interferometer ($88.5\pm1.2\%$), respectively. {\bf i-k} represent the cases of output photons of $a4$ and $b4$ in $\A$-type interferometer ($98.0\pm0.3\%$), $a4$ and $b$ in $\B$-type interferometer ($94.4\pm0.3\%$), $a$ and $b4$ in $\mathsf{AB}$-type interferometer ($94.8\pm0.9\%$), respectively. {\bf d, h} and {\bf l} represent the corresponding theoretical predictions.}\label{fig:Expre}
\end{figure*}

\begin{figure}[t]
\centering
\includegraphics[clip=true, width=\linewidth]{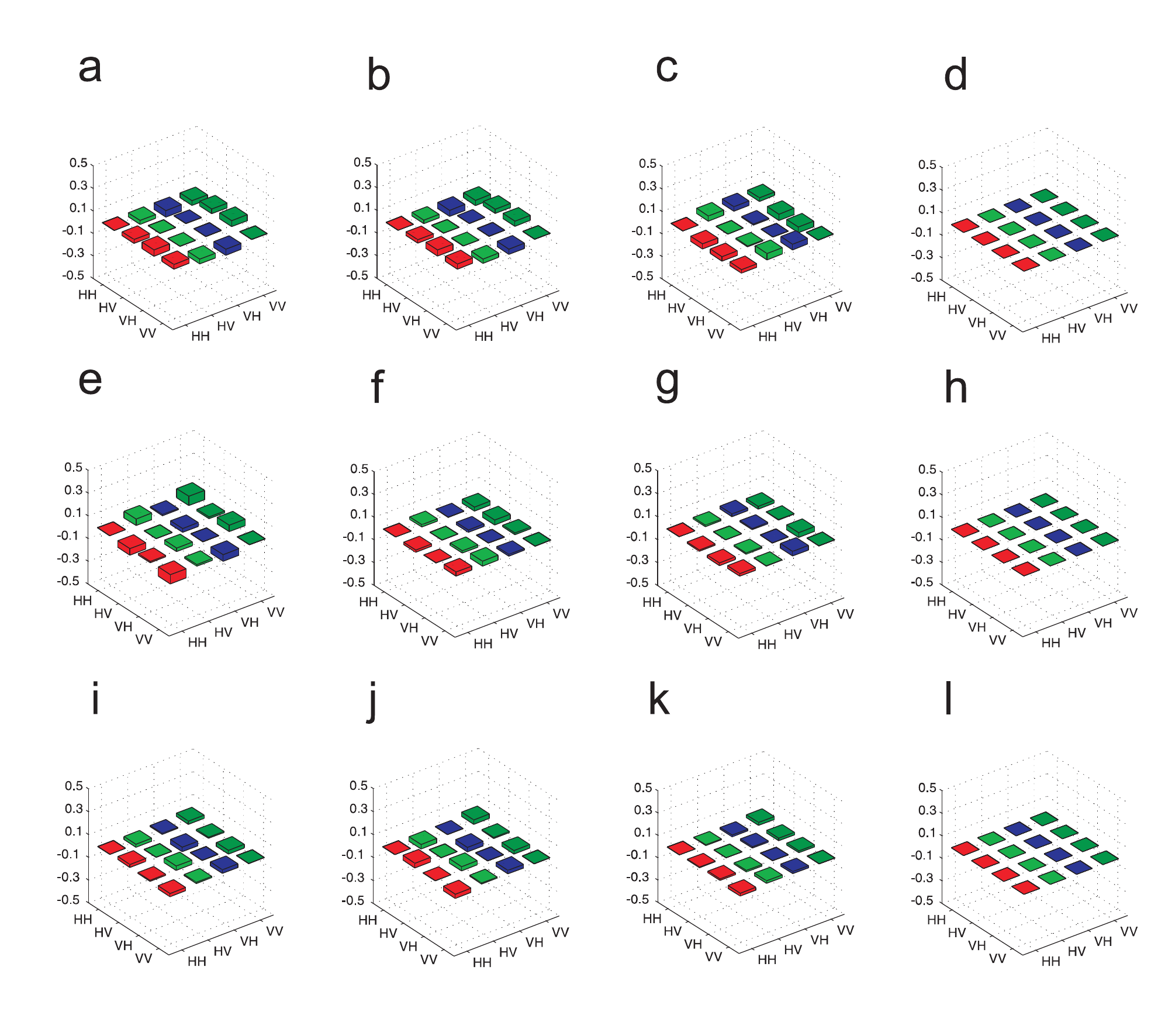} %
\caption{(Color online). Imaginary parts of density matrices of the finial output photons for the case with input state of $\frac{1}{\sqrt{2}}(\ket{HH}+\ket{VV})$. {\bf a-c} represent the cases of initial state, output photons of $a3$ and $b$ in $\E$-type interferometer, output photons of $a$ and $b3$ in $\mathsf{AB}$-type interferometer without CNOT gate, respectively. {\bf e-g} represent the cases of output photons of $a3$ and $b3$ in $\A$-type interferometer, $a3$ and $b$ in $\B$-type interferometer, $a$ and $b3$ in $\mathsf{AB}$-type interferometer, respectively. {\bf i-k} represent the cases of output photons of $a4$ and $b4$ in $\A$-type interferometer, $a4$ and $b$ in $\B$-type interferometer, $a$ and $b4$ in $\mathsf{AB}$-type interferometer, respectively. {\bf d, h} and {\bf l} represent the corresponding theoretical predictions.}\label{figs:Expim}
\end{figure}

	\end{appendix}
\end{document}